\documentclass[conference,letterpaper]{IEEEtran}

\usepackage{cite}
\usepackage{algorithm2e}
\usepackage{amsmath,amssymb,amsfonts}
\usepackage{algorithmic}
\usepackage{graphicx}
\usepackage{textcomp}
\usepackage{xcolor}
\usepackage{rotating}
\usepackage{multirow}
\usepackage[english]{babel}
\usepackage{float}
\usepackage{mathtools,eqparbox}
\usepackage[utf8]{inputenc} 
\usepackage[T1]{fontenc}
\usepackage{url}
\usepackage{ifthen}
\usepackage{cite}
\usepackage{graphicx} 
\usepackage[utf8]{inputenc}
\usepackage{amsmath}
\usepackage{amssymb}
\usepackage{amsthm}

\newcommand{\indices}[2]{{
  \begin{array}{@{}r@{}}
    \scriptstyle #2~\smash{\eqmakebox[ind]{$\scriptstyle\rightarrow$}} \\[-\jot]  
    \scriptstyle #1~\smash{\eqmakebox[ind]{$\scriptstyle\downarrow$}}
  \end{array}}}
\newtheorem{theorem}{Theorem}[section]

\newtheorem{lemma}[theorem]{Lemma}
\begin{document}
\title{Camouflage Adversarial Attacks on Multiple Agent Systems} 


\author{%
  \IEEEauthorblockN{Ziqing Lu}
  \IEEEauthorblockA{
                    University of Iowa\\
                    Iowa City, IA, USA\\
                    Email: ziqlu@uiowa.edu}
  \and
  \IEEEauthorblockN{Guanlin Liu}
  \IEEEauthorblockA{
                    University of California, Davis\\
                    Davis, CA, USA\\
                    Email: glnliu@ucdavis.edu}
                    \and\IEEEauthorblockN{Lifeng Lai}
  \IEEEauthorblockA{
                    University of California, Davis\\
                    Davis, CA, USA\\
                    Email: lflai@ucavis.edu}
                    \and
  \IEEEauthorblockN{Weiyu Xu}
  \IEEEauthorblockA{
  University of Iowa\\
                    Iowa City, IA, USA\\
                    Email:weiyu-xu@uiowa.edu}
}


\maketitle


\begin{abstract}
The multi-agent reinforcement learning systems (MARL) based on the Markov decision process (MDP) have emerged in many critical applications. To improve the robustness/defense of MARL systems against adversarial attacks, the study of various adversarial attacks on reinforcement learning systems is very important.  Previous works on adversarial attacks considered some possible features to attack in MDP, such as the action poisoning attacks, the reward poisoning attacks, and the state perception attacks. In this paper, we propose a brand-new form of attack called the camouflage attack in the MARL systems. In the camouflage attack, the attackers change the appearances of some objects without changing the actual objects themselves; and the camouflaged appearances may look the same to all the targeted recipient (victim) agents. The camouflaged appearances can mislead the recipient agents to misguided actions. We design algorithms that give the optimal camouflage attacks minimizing the rewards of recipient agents. Our numerical and theoretical results show that camouflage attacks can rival the more conventional, but likely more difficult state perception attacks. We also investigate cost-constrained camouflage attacks and showed numerically how cost budgets affect the attack performance.
\end{abstract}

\section{Introduction}
Single-agent and multiple-agent reinforcement learning (RL) algorithms are used in many safety or security related applications, such as autonomous driving\cite{ad}, financial decisions\cite{fd}, recommendation systems\cite{rs}, and also in drones' and robots' algorithms \cite{drones}. It is thus essential to develop trustworthy systems before their real-world deployment. Studying the potential adversarial attacks on RL systems and evaluating the worst-case performances of RL agents under these attacks can help us limit the damage imposed by adversarial parties, defend against adversarial attacks, and therefore build more robust and secure RL systems.

Adversarial attacks and defenses against these attacks for single-agent RL systems have been relatively well studied so far\cite{davis-sa, enmy, enchat, unlearner, vul-mech, defen1,defen2,defen3,defen4}, but adversarial attacks on multi-agent learning are still not well understood. In MARL, the model can still be based on the Markov Decision Process (MDP), but multiple players are playing in the Markov game (MG), interacting with the environment, and the environment dynamics change by the joint action of all agents. The increasing complexity of settings potentially makes MARL systems more fragile or makes it harder to analyze their robustness. New methods are introduced especially for improving/evaluating the performance of MARL systems, and evaluating worst-case adversarial attacks on MARL systems\cite{vlearning, davis-ma, multiql}. For example, \cite{vlearning} proposed a decentralized algorithm: V-learning that only scales with the maximum number of actions of one agent. In \cite{davis-ma}, the authors used reward loss and cost functions to evaluate the efficacy of adversarial attacks on MARL systems.

In terms of the types of adversarial attacks on MARL, most proposed adversarial attacks only consider recipient (victim) agents' properties to attack, for example, the action poisoning attacks, the reward poisoning attacks, the state poisoning attacks, the environmental attacks, or the mixed attacks  \cite{att-costsig, davis-ma, ada-rpa,davis-sa, poli-induc, batch-rl, spaold, envattack1,envattack2}. These attacks either directly change the features of agents, i.e., actions, rewards, or states of the MDP, or perturb the interactions between the agents' actions and the environments. In \cite{spa1, spa2}, the authors proposed a form of state perception (observation) attack in deep reinforcement learning, in which attackers confuse agents with delusional states instead of changing their actual states during the game. In \cite{cc-ma}, the authors addressed the state perception attacks with cost constraints in a multi-agent system.

In this paper, we propose a new form of adversarial attack on MARL system: the camouflage attack. During the camouflage attack, instead of directly changing recipients' properties, the attackers change the appearances of some objects they can control or even the appearances of attackers themselves. After the camouflage attack, all the recipient agents potentially observe the same camouflaged objects' features so that they are misled to misguided decisions in the MG. The camouflage attacks are different from the state perception attacks in two ways: 1) the camouflage attack does not directly change the measurements of each recipient agents, but instead change the appearances of the objects the attackers can control thus changing the measurements of the victim agents indirectly; 2) in the camouflage attack, the perceptions of different recipient agents cannot be freely manipulated as in state perception attacks: the confusions of the recipient agents come from observing the same camouflaged objects and thus are the same or correlated.  In addition, in camouflage attacks, the underlying true states of the camouflaged objects are not changed, and what are changed are only the appearances of the camouflaged objects. For example, camouflaged robot examples include stealthy invisible airplanes that can evade the detection of regular radars: they are actually in the air but are ``camouflaged'' to be invisible. As another example, one can camouflage unimportant objects into fake ``valuable'' targets so that enemy robots spend precious resources on attacking these fake targets. 

We design between-step dynamic programming to achieve the optimal camouflage attack in multiple agent systems. Our numerical results showed that the camouflage attack can significantly reduce the total reward gained by the recipient agents. We also numerically showed similar or comparable performances between the new, cheaper camouflage attack and the conventional but likely more difficult-to-achieve state perception attack. We proved that, under certain conditions, the camouflage attack can achieve a similar goal as the state perception attack.

There have been only very limited ideas of camouflage attacks studied from the perspective of dynamic systems, except for \cite{cam-ds} which discussed essentially state perception attack for single-victim-agent dynamic systems even though the terminology ``camouflage'' is used. For non-dynamic systems, some works discussed improving the detection of camouflaged attacks in deep learning models \cite{cam-dl, cam-dl2}. 

The rest of the paper is organized as follows. In Section II, we introduce the camouflage attack model. In Section III, we consider the budget-constrained scenarios with the camouflage attacks. In Section III, we analyze the performance of camouflage attacks and compare that with the performance of state perception attacks. In section IV, we numerically evaluate the performance of the proposed attack algorithms.

\section{Problem Formulation}
In the considered MDP environment, all the agents are divided into two opposite groups, the attacker group $M$, and the recipient agent group $N$, with $|M| = m, \ |N| = n.$ 

The MDP environment for the recipient agents can be described by the 5-element tuple: $(\{\mathcal{S}_i\}_{i=1}^n, \{\mathcal{A}_i\}_{i=1}^n,$ $T, P, \{R_i\}_{i=1}^n)$, where $\mathcal{S}_i$ is the state space of the $i$-th recipient agent with $|\mathcal{S}_i|=S_i$, $\mathcal{A}_i$ is the action space for the $i$-th recipient agent with $|\mathcal{A}_i| = A_i$, $T$ is the number of time steps in MARL and $T$ is finite. Our time index $t$ starts from 0 to $T$. We refer to time step $t$ as the time interval starting from time index $t-1$ and ending at time index $t$, where $1\leq t \leq T$. We let $P_{i,t}:\mathcal{S}_i \times \mathcal{A}_i \times \mathcal{S}_i \rightarrow [0,1]$ be the transitional probability of the $i$-th recipient agent at time index $t$, and ${R}_{i,t}$ represents the reward function of the $i$-th recipient agent at time index $t$. We let $\textbf{a}_t := (a_1, a_2, \dots, a_n)$ denote the {joint action} of all the $n$ recipient agents at time step $t$, and let $\textbf{s}_t := (s_1, s_2, \dots, s_n)$ denote the joint state of all the $n$ recipient agents at a time index 
$t$. We define the optimal policy of the $i$-th recipient agent at time index  $t$ as $\pi_{i,t}^{\star}(s_{i,t}) = a_{i,t}^{\star}$. For simplicity, we assume that all the recipient agents share the same state space and the same action space. We assume that the optimal policy $\pi_{i}^{\star}(s_{i,t})$ is the same for every agent $i$, denoted as $\pi_t^{\star}(s_{t})$. 

 We assume both the attack and recipient agent groups can monitor the underlying MARL algorithms of the recipient agents, and therefore both groups know the optimal policies $\pi^{\star}_{i,t}$ of every recipient agent $i$ at time step $t$. However, the recipient agents are unaware of the existence of attackers or their attacks. The $m$ attackers perform state {perception} attacks by disturbing recipient agents' observations of their true states. For a recipient agent $i$ at time index $t$,  we let $s_{a,t,i}$ denote the true state the agent $i$ is actually in and let $s_{d,t,i}$ denote the delusional state that the agent $i$ thinks it is in.

The recipient agents are selfish in the game, aiming to maximize their own rewards obtained during the $T$ time steps. The attackers instead aim to minimize the total expected rewards of all the recipient agents during the $T$ time steps. 

There are two phases of play during one time step $t$. In the first phase, from time index $t-1$ to $t-0.5$, the attackers attack to make each recipient agent $i$ ($1\leq  i \leq n$) think it is in a delusional state ${s}_{d,t-0.5, i}$. In the 2nd phase, after the attack, from the time index $t-0.5$ to $t$, each recipient agent $i$ moves to $s_{a,t,i}$ according to its optimal policy $\pi_{i,t-0.5}^{\star}(s_{d,t-0.5,i})=\pi_{i,t-1}^{\star}(s_{d,t-0.5,i})$, in which $s_{d,t-0.5,i}$ is agent $i$'s delusional state at time index $t-0.5$, and obtains its corresponding reward $R_{i,t}(s_{a,t-1,i}, s_{a,t,i})$.

We are interested in finding the optimal attack strategy of attackers for each time step $t$ ($1\leq t \leq T$).

\textbf{Camouflage attack}: 
The $m$ attackers can change the appearances of some objects that they control during the game at every time step $t$. Mathematically, suppose that we have a random variable $X_{t-0.5}$ which represents the true status of an object at time index $(t-0.5)$. If it is not camouflaged, the appearance, denoted by $Y_{t-0.5}$, of this object is just equal to $X_{t-0.5}$, namely $Y_{t-0.5}=X_{t-0.5}$. {The camouflage attack changes the appearance $Y_{t-0.5}$ to some other value, for example, $Y_{t-0.5}=g(X_{t-0.5})$ where $g(\cdot)$ is a camouflage function. Then the observation of $X_{t-0.5}$ at time index $(t-0.5)$ from the perspective of agent $i$ is given by $s_{d,t-0.5,i}=h_i(g(X_{t-0.5}))=h_i(Y_{t-0.5})$, where $Y_{t-0.5}$ is the changed appearance of the object, and $h_i$ is the observation function of recipient agent $i$ ($h_i$ can be a function giving random outcomes, for example, due to noises). 

These camouflaged objects will fool the recipient agents into delusional state $s_{d,t-0.5,i}$ for each agent $i$. Different from the state perception attacks, these delusional states  $s_{d,t-0.5,i}$ have to be correlated or even exactly the same across different recipient agents: this is because camouflage attacks make the recipient agents observe the same camouflaged objects. In state perception attacks \cite{spa1,spa2}, the attackers can instead fool different recipients into very different delusional states.  

\section{Cost constrained camouflage attacks}
\label{sssec:al-time-cc}

In practice, sometimes the attackers have attack budgets that must be spent by the end of each time step. For example, the resources used by attackers are provided by constantly-energy-harvesting systems over time steps, and the budget for each time step is constrained by the battery volume.  We call this scenario an ``instant cost constrained case''. Within each time step $t$ (namely between time index $t-1$ and $t$), 
all attackers share a budget $B$, and this budget $B$ can only be spent during that single time step: the leftover resources cannot be carried over to the next time step $t+1$ or there is no need to carry over the leftover resources to the next time step because of budget refill. Once we get to time step $t+1$, the shared budget $B$ will be refilled (say, to $B$). We would like to find out how to optimally allocate the total resources to each attacker $j$ for performing the camouflage attack in each time step, while satisfying the instant cost constraint and minimizing the recipient agents' total rewards. To simplify our presentations, we consider the budgets are used to camouflage the attackers themselves.

We design an integration of between-step dynamic programming and within-step static constrained optimizations to compute the optimal attack strategy. During each time step $t$, for each possible actual state vector $\textbf{s}_{a,t-1}$, we use a static standalone optimization program to determine the optimal allocation of attackers' budgets on camouflages.  Between different time steps, we use dynamic programming to account for the state transitions and expected rewards.

We work backward from time index $t=T$ and initialize value function $V_{T}^{\star}(\sigma_T)=0$ for each dynamic programming state (DPS) $\sigma_{T}$ (a DPS state includes all recipient and attacker agents' actual conditions, and also the conditions of camouflaged objects), and the subscript represents time index. Suppose that we have already computed $V_{t+0.5}^{\star}(\sigma_{t+0.5})$ for every DPS state $\sigma_{t+0.5}$.  We then proceed to compute the optimal attack policy during time step $t$ (essentially from $t$ to $t+0.5$) and also $V_{t}^{\star}(\sigma_t)$ for every DPS state $\sigma_t$. During time step $t$, we let $b_{j} \in \mathbb{R}$ be the amount of resources attacker $j$ spends on its camouflage attack. The constraints on $b_j$ are such that the total spending of all attackers cannot exceed $B$. To represent formulas concisely, we stack the $b_{j}$'s to form a $m$-dimensional vector $\mathbf{b}$ called the attack allocation vector. Under the attack allocation vector $\mathbf{b}$, we denote the probability that the recipient agents' state will transit to $\sigma_{t+0.5}$ as $P(\mathbf{b},\mathbf{s}_{a,t}, \sigma_{t+0.5})$, where $\mathbf{s}_{a,t}$ is the true states of all the recipient agents at time index $t$. 
This probability must be between 0 and 1. Based on the principles of dynamic programming, we want to optimize $b_{j}$'s to minimize the total expected rewards received by the agents from time step $t$ to $T$. Thus, under a specific true state vector $\mathbf{s}_{a,t}$, the objective function for attackers to minimize is the expected total reward of all the recipient agents from step $t$ onward to step $T$. 

\subsubsection{Within-step static constrained optimization problem}
\label{sec:staticoptimization}
Suppose that the DPS has $Q$ possible values at time index $t+0.5$ conditioned on the true states are $\mathbf{s}_{a,t}$, the optimal attack under the ``instant cost constrained'' case at a single time step $t$ can be formulated as the following within-step static constrained optimization problem:
\begin{align}
\label{static}
\min \sum_{k=1}^{Q} & P(\mathbf{b},\mathbf{s}_{a,t}, \sigma_{t+0.5}^k) V_{t+0.5}^{\star}(\sigma_{t+0.5}^k)\\
\text{subject to} &\sum_{j=1}^{m} b_{j} \leq B, \nonumber\\
& P(\mathbf{b},\mathbf{s}_{a,t}, \sigma_{t+0.5}^k) \leq 1, \ \forall k  \nonumber \\
& -P(\mathbf{b},\mathbf{s}_{a,t}, \sigma_{t+0.5}^k) \leq 0, \  \forall k \nonumber\\
& b_{j} \geq 0, \  j = 1,\dots,m, \nonumber
\end{align}
where $\sigma_{t+0.5}^k$ is the $k$-th DPS at time index $t+0.5$.

Depending on the physical nature of the attacks, we can model the probability $P(\mathbf{b},\mathbf{s}_{a,t}, \sigma_{t+0.5}^k)$ as a function of $\mathbf{b}$. In one particular model considered in the paper, for each attacker $j$, the probability that it can change the appearance of the object it controls is $\max\{b_j/C_t(x_j,y_j),1\}$, where $C_t(x_j, y_j)$ are constants representing how hard it is for the attacker $j$ to camouflage the appearance of $x_j$ as $y_j$. In our numerical results, we take $C_t(x_j,y_j) = d(s^{\dagger}_{a,j}, s^{\dagger}_{d,j}) + \epsilon$ where $\epsilon$ is a positive constant and $d(s^{\dagger}_{a,j}, s^{\dagger}_{d,j})$ is the distance between the real position of the attacker $j$ and the target camouflaged position the attacker $j$ chooses. Namely, if we assign more budget to attacker $j$, and if the target camouflage position is closer to its actual position, it is more likely that attacker $j$ can change the objects to the targeted appearances.

\subsubsection{Between-step dynamic programming}
After solving \eqref{static}, we take the optimal value of its objective function as $V_t^{\star}(\sigma)$, using which we continue to calculate $V_{t-0.5}^{\star}(\cdot)$ as follows. For each DPS $\sigma_{t-0.5}$, we update $V_{t-0.5}^{\star}(\sigma_{t-0.5})$ as 
$$\small \sum_{\sigma_t} P(\sigma_{t} | \sigma_{t-0.5}, \textbf{a}^{\star}_{t-0.5}) ( V_t^{\star}(\sigma_t)  + R(\sigma_{t-0.5},\sigma_t)).$$
After updating $V_{t-0.5}^{\star}(\cdot)$, again we will use another static optimization formulation in Section \ref{sec:staticoptimization} to calculate $V_{t-1}^{\star}(\cdot)$. In this way, we perform this static optimization-dynamic programming cycle recursively until we calculate all the $V_{t}^{\star}(\cdot)$ backward from $t=T$ until $t=0$. 

\section{Performance analysis of camouflage attacks}
Camouflage attack is arguably a more practical form of adversarial attacks since it only requires the attackers change the appearances of the objects the attackers directly control. So different victims will have correlated or the same observations of these camouflaged objects.  In contrast, the optimal state perception attacks would require the attackers to change the observations of different victims to possibly different delusions. The following analytical results bound the gaps between the camouflage attacks and the state perception attacks. In this section, we assume that different victims have the same observations of the camouflaged objects and we do not impose cost constraints on the attacks.   We start with a lemma about imposing an equality constraint on the optimization variables.  

\begin{lemma}\label{lemma1}
    Consider $n$ functions $\{f_i\}_{i=1}^n$, where $i=1,2,~...,~n$,  and the following two optimization problems:
    \begin{align}\label{cam}
        \min_{x_1,x_2,\dots ,x_n} &\sum_{i=1}^n f_i(x_i)\\
        \text{subject to } &x_1=x_2=\dots=x_n;\nonumber
    \end{align}
    and 
    \begin{align}\label{spa}
        \min_{x_1,x_2,\dots ,x_n} \sum_{i=1}^n f_i(x_i).
    \end{align}
    Let $x^{**}$ be the optimal solution of (\ref{cam}) and $o_1$ be the optimal objective value of (\ref{cam}). Let $(x_1^*, x_2^*, \dots, x_n^*)$ be the optimal solution of (\ref{spa}) and $o_2$ be the optimal objective value of (\ref{spa}). Assume that there exist constants $C_j$'s, $j=1,2,~...,~n$, such that for every $j$, 
\begin{align*}
    \sum_{i=1,i\neq j}^n (f_i(x_j^*) - f_i(x_i^*)) \leq C_j.
\end{align*}
Then we have $ o_2 \leq o_1 \leq o_2 + \min_{j}{\{C_j\}}.$
\end{lemma}
\begin{proof}: Because (\ref{cam}) has one additional constraint,  $o_2 \leq o_1$. Given an arbitrary index $j$, $j=1\dots n$, we have:
\begin{align*}
   o_1 &= \sum_{i=1}^n f_i(x^{**}) \leq \sum_{i=1}^nf_i(x_j^*).
\end{align*}
So for every $j$, we have 
\begin{align*}
    o_1-o_2&\leq \sum_{i=1}^nf_i(x_j^*) - \sum_{i=1}^nf_i(x_i^*)\\&=\sum_{i=1,i\neq j}^n (f_i(x_j^*) - f_i(x_i^*))\leq C_j.
\end{align*}
Therefore $o_2\leq o_1\leq o_2 + \min_{j}\{C_j\}$.
\end{proof}
\begin{theorem}\label{thm1}
Consider $m$ attacker and $n$ recipient agents, for one single time step $t$ (from time index $t-1$ to time index $t$).  Assume the recipients share the same state space $\mathcal{S}$, the same action space $\mathcal{A}$, the same probability transition matrices $P: \mathcal{S}\times \mathcal{A}\times \mathcal{S} \rightarrow [0,1]$, and the same reward function $R:\mathcal{S} \times \mathcal{S} \rightarrow \mathbb{R}$. Let the observation functions $h_i$ be identical for every agent $i$, so that the camouflaged observations are the same for every recipient agent $i$, i.e. $s_{d,t-0.5,i}$'s are equal. We assume that the optimal policy for each recipient agent is the same and the recipients work independently from each other. We use $\pi^*_{t}:\mathcal{S}\rightarrow \mathcal{A}$ to denote the shared optimal policy of a recipient agent at time step $t$. Within time step $t$, let the total rewards of all recipients gained under the optimal camouflaged attack be ${TR}^{ca}_t$ and the total reward gained under the optimal state-perception attack be ${TR}^{spa}_t$.

Assume that for every pair of two different recipient agents $(i,j)$, $i,j=1\dots n$, for every pair of actual states $(s_{a,t-1,i}, s_{a,t-1,j})$ of recipient $i$ and recipient $j$, the rewards gained for agent $i$ under delusional state perceptions at time step $t$ satisfy
\begin{equation*}
   \resizebox{0.97\hsize}{!}{$ER(s_{a,t-1,i}, \pi^*_{t}(s^*_{d,t-0.5,j})) - ER(s_{a,t-1,i}, \pi^*_{t}(s^*_{d,t-0.5,i}))\leq C_{ij},$} 
\end{equation*}
for some small constant $C_{ij}$, where $s^*_{d,t-0.5,j}$ and $s^*_{d,t-0.5,j}$ are the most-damaging delusional state perceptions that minimize the reward for agent $j$ and $i$ respectively, and $ER(s_{a,t-1,i}, \pi^*_{t}(s_{d,t-0.5,\cdot}))$ is the expected reward recipient agent $i$ will get using the policy corresponding to a delusional state perception $s_{d,t-0.5, \cdot}$.  Then  
$${TR}^{spa}_t\leq {TR}^{ca}_t\leq {TR}^{spa}_t + \min_j \sum_{i=1,i\neq j}^n\{C_{ij}\}.$$
\end{theorem}
\begin{proof}: We use Lemma \ref{lemma1} to prove Theorem \ref{thm1}. For each recipient agent $i$, we let the function $\{f_i\}_{i=1}^n$ be the expected reward recipient agent $i$ gets under its true states $s_{a,t-1,i}$  and agent $i$'s delusional observation $s_{d,t-0.5,i}$.   

In this setting, the variable $x_i$ in the optimization problems (\ref{cam}) and (\ref{spa}) is the delusional observation $s_{d,t-0.5,i}$ of the $i$-th recipient. In (\ref{cam}) which corresponds to the camouflage attack, $f_i(x_i) =ER(s_{a,t-1,i}, \pi^*_t(s_{d,t-0.5,i}))$ and we require $s_{d,t-0.5,i}$ to be equal across different agents $i$'s. In  (\ref{spa}) which corresponds to a ``free'' state perception attack, $f_i(x_i) = ER(s_{a,t-1,i}, \pi^*_t(s_{d,t-0.5,i}))$,  but we will not require  $s_{d,t-0.5,i}$ to be the same across different agents $i$'s. Because \begin{equation*}
    \resizebox{0.97\hsize}{!}{$ER((s_{a,t-1,i}, \pi^*_{t}(s^*_{d,t-0.5,j})) - ER(s_{a,t-1,i}, \pi^*_{t}(s^*_{d,t-0.5,i})) \leq C_{ij},$}
\end{equation*}
by applying Lemma \ref{lemma1}, we have 
${TR}^{spa}_t\leq {TR}^{ca}_t\leq {TR}^{spa}_t + \min_j \sum_{i=1,i\neq j}^n\{C_{ij}\}$.
\end{proof}

\section{Numerical Results}
We perform numerical results under various game settings with $T=5$. In each setting, we compared the total rewards of all recipients without attack, with the camouflage attack, and with the state perception attack. Results indicate that recipients gain significantly smaller rewards under the camouflage attacks compared to the case with no attacks. The reward gained under the state perception attacks is smaller, but not significantly, than the more practical camouflage attack. For cost-constrained camouflage attacks, as the attack budget increases, the gained reward becomes smaller. Our framework works for general $m$-attacker-$n$-recipient scenarios.
\subsection{Camouflage orientations  }
In the first experiment, there are 2 recipients and 2 attackers playing in the MG. The recipients share the same state space $\mathcal{S}$, which is a ring containing 3 different states: 0, 1, and 2. They also share the same action space $\mathcal{A}$, the probability transition $P$, and the same reward function $R$. The action space $\mathcal{A}$ is composed of three actions: go left, go right, and stay. For actions \textbf{left} and \textbf{right}, the recipient agent has a 0.8 probability of moving in the intended direction and a 0.2 probability of moving in the opposite direction. For \textbf{stay}, the recipient agent has a 0.8 probability of staying at the current state and a 0.1 possibility of moving to the right or left. The reward function $R(s_{t-1}, s_{t})$ assigns a fixed positive reward to the recipient agents, which is displayed in the table below:\\
\begin{equation*}
    \resizebox{0.39\hsize}{!}{$
  \begin{array}{c|ccccc}
    \indices{\text{t}}{\text{t-1}}
    & s_0 & s_1 & s_2\\
    \hline
    s_0 & 3.0 & 10.6 & 1.0 \\
    s_1 & 10.0 & 1.0& 0.0\\
    s_2 & 1.0 & 0.0 & 11.6
  \end{array}
$}
\end{equation*}

The attacks camouflage the orientation of the ring by rotating it counter-clockwisely for respectively 1 step, 2 steps, and 3 steps. The camouflaged orientation after a 3-step rotation is the same as the true orientation. For every attack, recipients' perceptions of their real positions are based on the camouflaged ring, as described in Figure 1. In Figure 2, we compared the expected global rewards of recipients for time index from 0 to 5, under camouflage attacks, state perception attacks, and without attacks.
\begin{figure}[htb]
\begin{minipage}[b]{0.48\linewidth}
  \centering
  \centerline{\includegraphics[width=4.5cm, height=4cm]{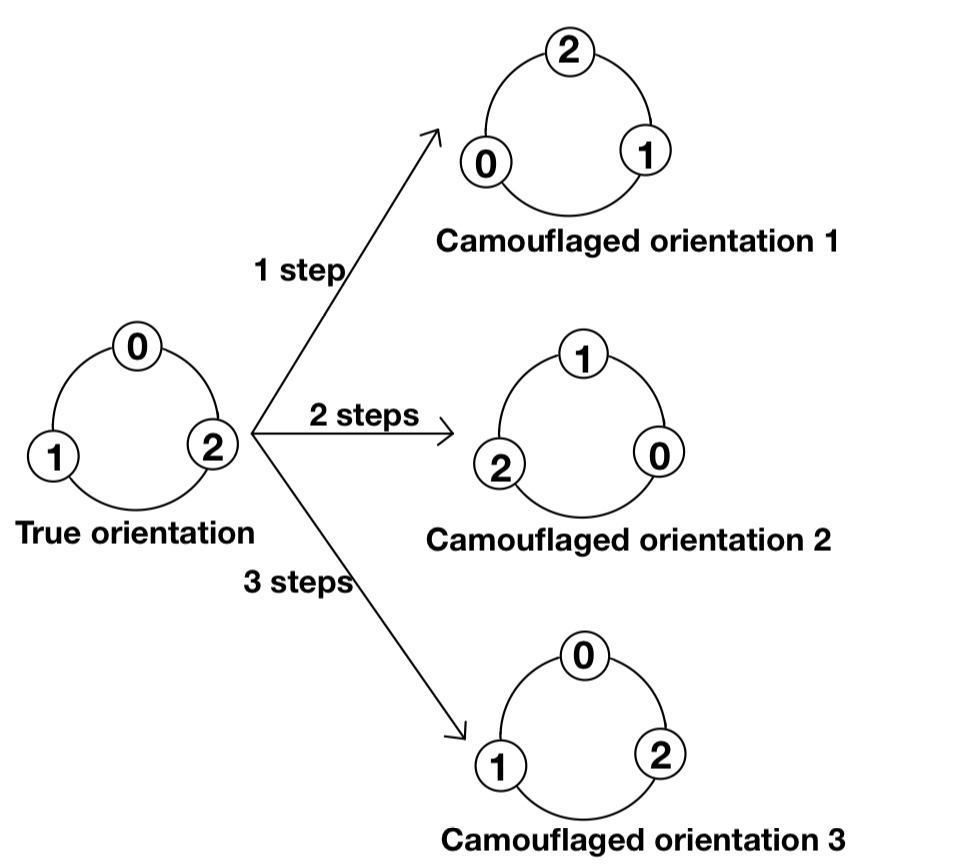}}
  \caption{Illustration: camouflage attacks on a ring.}
\end{minipage}
\hfill
\begin{minipage}[b]{0.48\linewidth}
  \centering
  \centerline{\includegraphics[width=4.5cm, height=4cm]{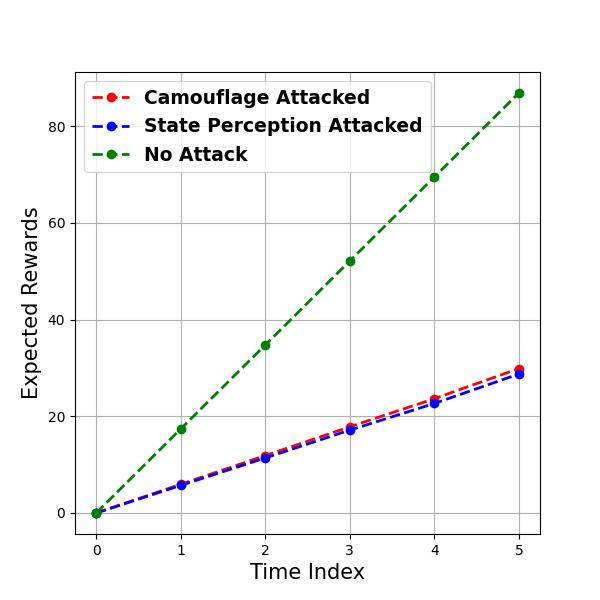}}
  \caption{Ring topology. Comparison of expected global rewards between free state perception attacks and camouflage attacks.}
\end{minipage}
\end{figure}
In this figure, the x-axis is the time index, and the y-axis is the total expected rewards the recipient agents gained from $t=0$ to the current time index. With camouflage attacks, the reward gained is 34.4$\%$ of that achieved without attack. With state perception attacks (where attackers can freely fool each recipient into desired delusional states), the reward is about 33.1$\%$ of that without attack.

\subsection{Camouflage attackers' real positions}
\label{sec:chessboard}
In this case, the recipients and the attackers move on a square $q \times q$ chessboard. The position of either a recipient or an attacker can be denoted as $(i,j)$, where $i$ ($0\leq i \leq q-1$) records the row index, and $j$ ($0\leq j \leq q-1$) records the column index on the chessboard. We made 2 experiments with a square chessboard: the first one has $q=3$, and the second one has $q=2$. For each attacker, its position is fixed and it can only attack if any recipient moves to its location. Neither attackers nor recipients can move beyond the boundaries of the chessboard. 

The recipients have the same state space $\mathcal{S}$, the same action space $\mathcal{A}$, and the same reward function $R$. Recipients can move up, down, right, or left. The attackers' positions cannot move during the game. For simplicity, the probability of recipients moving along the indicated direction of an action $a\in \mathcal{A}$ is set to be 1. The reward ${R}(s_{t-1}, s_{t})$ for entering all possible chessboard positions is set to be $5.0$ except for $(0,1)$, whose reward is $10.0$. If any recipient enters any of the attacker's positions, the reward is set to $1.0$. For example, with $q=3$, if the fixed positions of the attackers are at $(1,1)$ and $(2,1)$, then the reward function can be displayed in the following table. The $i$, $j$ here means the indices of the square the recipient is entering:\\
\begin{equation*}
     \resizebox{0.34\hsize}{!}{$
  \begin{array}{c|ccccc}
    \indices{\text{j}}{\text{i}}
    & 0 & 1 & 2\\
    \hline
    0 & 5.0 & 10.0 & 5.0 \\
    1 & 5.0 & 1.0& 5.0\\
    2 & 5.0 & 1.0 & 5.0\\
  \end{array}
$}
\end{equation*}
In the first experiment, 3 recipients and 2 attackers played on a 3$\times$3 chessboard, and in the second experiment, 2 recipients and 1 attacker played on a 2$\times$2 chessboard.  In Figures 3 and 4, we compared the expected global rewards of recipients for time index from 0 to 5, under camouflage attacks, state perception attacks, and without attacks.
\begin{figure}[htb]
\begin{minipage}[b]{0.48\linewidth}
  \centering
  \centerline{\includegraphics[width=4.5cm, height=3.9cm]{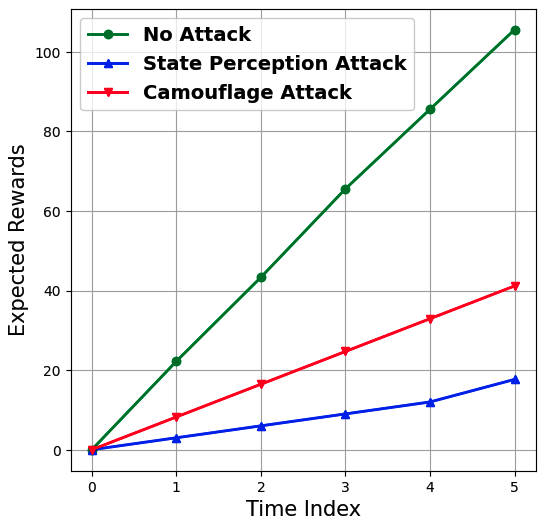}}
  \caption{$3\times 3$ chessboard. Comparison between state perception attack and camouflage attack, with fixed attackers at (1,1) and (2,1), 3 recipients and 2 attackers.}
\end{minipage}
\hfill
\begin{minipage}[b]{0.48\linewidth}
  \centering
  \centerline{\includegraphics[width=4.5cm, height=3.9cm]{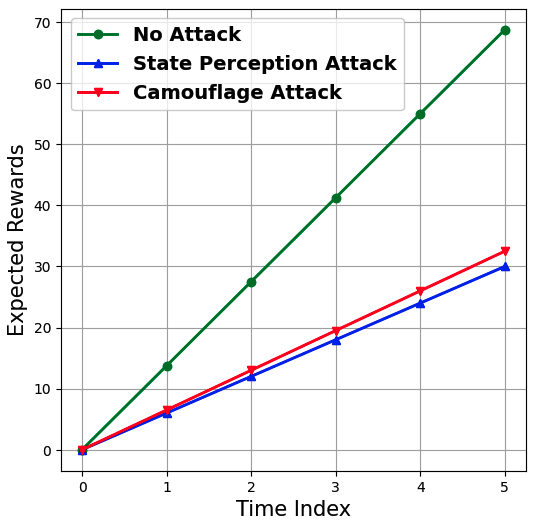}}
  \caption{$2 \times 2$ chessboard. Comparison between state perception attack and camouflage attack, 2 recipients and 1 attacker.}
\end{minipage}
\label{fig:res}
\end{figure}

In Figure 3, with the fixed attackers at (1,1) and (2,1) on the chessboard, the global reward gain achieved under camouflage attacks is 39.0$\%$ of the reward achieved without attack. The reward gain under state perception attacks is roughly 16.7$\%$. In Figure 4, the expected global gained reward over all real attackers' positions after 5 camouflage attacks is 47.3$\%$ of the case without attack, and the expected gained reward after 5 state perception attacks is 43.6$\%$ of the case without attack.

\subsection{Cost constrained camouflage attacks}
With 3 recipients and 2 attackers in the same $3 \times 3$ chessboard setting as \ref{sec:chessboard}, we add the cost constraint to attackers at every time step when they perform a camouflage attack. The cost of every attacker is the distance between its real position and the target camouflage position. We define the distance between two positions as the sum of their row and column index absolute differences. Within every time step $t$, the shared budget is refilled to a fixed budget. We choose the following sequence of fixed budgets $\{1,2,3,4,6,12\}$ for tests. In Figure 5, we compare the expected global reward gain under different budgets. It turns out that the higher the budget, the fewer the reward gains. When the budget reaches 6, the performance of the cost constrained camouflage attack is the same as the optimal camouflage attack.
\begin{figure}[h]
\centering
  \centerline{\includegraphics[width=7cm,height=4.4cm]{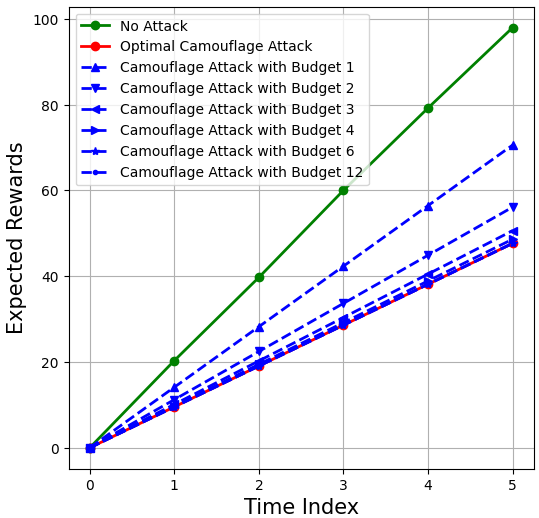}}
  \caption{Comparison between camouflage attacks with different budgets.}

\end{figure}






\bibliographystyle{IEEEbib}
\bibliography{refs}

\end{document}